\documentclass{article}    

\usepackage{etoolbox}
\newtoggle{full}
\toggletrue{full}    
\iftoggle{full}{
\usepackage[margin=1in]{geometry}
\usepackage{amsthm, amssymb}
\newtheorem{theorem}{Theorem}[section]
\newtheorem{lemma}[theorem]{Lemma}

}{

\usepackage{makeidx}

}%
\usepackage{amsmath,amssymb}
\usepackage{multicol}
\newcommand{\proofbox}{\hfill$\square$}

\title{LP-based Covering Games with Low Price of Anarchy}

\iftoggle{full}{
\author{Georgios Piliouras\thanks{School of Electrical \& Computer Engineering, Georgia
  Institute of Technology, Atlanta, GA. Supported by AFOSR projects
  FA9550-09-1-0538. {\tt
    georgios.piliouras\char'100 ece.gatech.edu}}\and
Tom\'a\v s Valla\thanks{Institute for Theoretical Computer Science, Faculty of Mathematics \&
Physics, Charles University, Prague, Czech Republic. Supported by the GAUK Project 66010 and
by ITI, under grant P202/12/G061.
{\tt valla\char'100 kam.mff.cuni.cz}
}
\and
L\'aszl\'o A. V\'egh\thanks{College of Computing, Georgia Institute of Technology.
Supported by NSF Grant CCF-0914732.
{\tt lvegh\char'100 cc.gatech.edu} }}
}
{
\author{Georgios Piliouras\inst{1}\thanks{Supported by AFOSR projects FA9550-09-1-0538.}
\and Tom\'a\v s Valla\inst{2}\thanks{Supported by GAUK Project 66010 and ITI grant number P202/12/G061.}
\and  L\'aszl\'o A. V\'egh\inst{3}\thanks{Supported by NSF Grant CCF-0914732.}
}
\institute{%
School of Electrical \& Computer Engineering, Georgia
  Institute of Technology, Atlanta, GA. \email{georgios.piliouras@ece.gatech.edu}
\and
Institute for Theoretical Computer Science, Faculty of Mathematics \&
Physics, Charles University, Prague, Czech Republic.
\email{valla@kam.mff.cuni.cz}
\and
College of Computing, Georgia Institute of Technology, Atlanta, GA. \email{lvegh@cc.gatech.edu}
}

}

\def\qed{\ifhmode\unskip\nobreak\fi\hfill
  \ifmmode\square\else$\square$\fi}
\def\R{{\mathbb R}}

\def\C#1{{\mathcal #1}}

\let\epsilon=\varepsilon

\def\I{\it\aftergroup\/}

\def\text#1{\hbox{#1}}

\begin{document}

\maketitle

\begin{abstract}
We present a new class of vertex cover and set cover games. The
price of anarchy bounds 
match the best known constant factor approximation guarantees for the
centralized optimization problems   for linear and also
for  submodular costs -- in contrast to all previously studied covering
games, where the price of anarchy cannot be bounded by a constant (e.g. \cite{Cardinal06,Cardinal10,Escoffier10,buchbinder08,Balcan11}).
In particular, we describe a vertex cover game with a price of anarchy of $2$.
The rules of the games capture the structure of the linear programming
relaxations of the underlying
optimization problems, and our
bounds are established by analyzing these relaxations.
Furthermore, for linear costs we exhibit linear time best response
dynamics that converge to these almost optimal Nash equilibria. These
dynamics mimic the classical greedy approximation algorithm of
Bar-Yehuda and Even \cite{Baryehuda81}.
\end{abstract}

\section{Introduction}

Combinatorial optimization has for several decades  dictated the landscape of algorithm design.
The extent of its impact can be appreciated by the fact that almost by default the main
judging criterion of a polytime algorithmic solution is the approximation guarantee it offers, regardless of other parameters that may affect the applicability of the solution in practice (simplicity of implementation, robustness to input errors, etc.)

One such limiting assumption 
is the existence of an omnipotent centralized authority that has access to all
the relevant information and has the power to enforce any solution of its choice. Over the last decade, the soundness of such assumptions has increasingly come into question
following a number of paradigm-shifting  socioeconomic events such as the rapid rise of the Internet, the painful realization of the extent of inter-connectivity of the
global economy as well as the emergence of global
sustainability concerns.

Algorithmic game theory strives for global optimization
in such decentralized settings that consist of self-interested individuals.
In these more challenging scenarios, tractability can be compromised along
two largely independent axes: due to individual incentive issues or due to computability issues.

\iftoggle{full}{\smallskip}{}

\noindent{\bf Price of Anarchy.} The competition between individual incentives and social optimality is of fundamental concern in distributed
systems as it can lead to highly inefficient outcomes.
The price of anarchy literature \cite{poa} examines exactly what are the worst case repercussions of such a policy. Formally, price of anarchy is defined as the maximal ratio between the
social cost of a Nash equilibrium and that of the global optimal configuration. Intuitively, a low price of anarchy implies that upon converging to a socially stable outcome, the quality of the acquired solution is almost optimal from a central optimization perspective.

Unfortunately, in many cases of interesting games the price of anarchy is
prohibitively high. Vertex cover, due to its prominent position within combinatorial optimization, has been studied in the context of game theory from different approaches, all of which so far have shared this limiting
characteristic.

Specifically,  Cardinal and
Hoefer in \cite{Cardinal06} define a vertex cover game where the edges of a network are
owned by $k$ agents.
An agent's goal is to have  each of his edges
supplied by a service point at least one of its endpoints. There
is a cost $c(v)\ge 0$ associated to building a service point at vertex
$v$. The strategy of an agent is a vector consisting of offers to the
vertices. Service points will be installed at vertices where the total
offer exceeds the cost of the vertex.
Similar games are defined by Buchbinder et al. \cite{buchbinder08} and by
Escoffier et al. \cite{Escoffier10} for the more general set cover problem.

A different approach was followed by Balcan et
al. \cite{Balcan11}. Here the agents are the vertices of the graph, and
their strategies are deciding whether they open a service point.
If opening a service point, vertex $v$ incurs a cost $c(v)$. If
he decides not to open, he has to pay a penalty for all edges incident
to $v$ whose other endpoints are uncovered.

The price of anarchy is 
$\Theta(k)$ in
\cite{Cardinal06} and $\Theta(n)$ in \cite{Balcan11}. Indeed, if
the underlying network is a star, and each edge is owned by a
different agent in the first case, we get Nash equilibria with all
leaves being service points. These guarantees are significantly
worse than the ones available in the centralized setting, where
simple factor $2$-approximation algorithms exist.

In contrast, in our paper, we shall present a simple vertex cover game with a price of anarchy $2$.
As in \cite{Balcan11}, the agents are the vertices, and the regulations delegate the responsibility of covering every
edge of the network to its two endpoints: both incur a high penalty if 
the edge is left uncovered.
The difference from the setting
of \cite{Balcan11} is that those who  open a service
point can demand compensation from their neighbors.
This is justified since if $u$ opens a service point, every neighbor
$v$ benefits from this as the common responsibility of covering
$uv$ is taken over by $u$.

In the description, we use 
intuitive terminology
of a Mafia (service points) which provides ``security'' (covers
edges).
The vertices may choose to join Mafia or to
remain civilians. Each edge of the graph has to be ``secured'', that
is, at least one endpoint must be in Mafia. 
For agent $v$, there is an initial cost $c(v)$ to join  Mafia. 
Mafiosi can collect
ransoms as the price of security of the incident edges:
if a vertex $v$
chooses to be a mafioso, his strategy also includes a ransom vector, so
that the total ransom he demands from his neighbors is $c(v)$. 
It is a one-shot
game and mafiosi can ransom
both their civilian and mafioso neighbors.

If $v$ is a civilian,  he has to pay to his neighbors in the
Mafia all ransom they demand. Furthermore, if there is an incident
uncovered edge $uv$, that is, $u$ is also a civilian, both of them have to pay a huge penalty.
In contrast, if $v$ is a mafioso, he has to pay $c(v)$ for joining, and
 he receives whatever he can collect from ransoms. However, mafiosi
ransomed excessively obtain a protected status: 
if the total demand
from $v$ is more than $c(v)$, he satisfies only a proportional fraction of the demands.
It is important to note that the payoff function is defined locally: besides his own strategy, the payoff of an agent
depends only on the strategies of
agents at distance at most $2$ from him (i.e. immediate neighbors and neighbors of neighbors).
Also note that if $M$ is a vertex cover, then the total utility of the
agents is $-c(M)$. Consequently,   an
optimal solution to the optimization problem gives a social optimum of
the game.

\iftoggle{full}{
Our approach avoids bad Nash equilibria that are possible in \cite{Cardinal06} and \cite{Balcan11}.
As an example, consider the vertex cover game on a star with all
vertices having cost 1.
In the models of
\cite{Cardinal06} and \cite{Balcan11} there exists a Nash equilibrium
where
the leaves form the vertex cover.
In our model, if all leaves are mafiosi, then all of them would demand
ransom from the central agent, who would then have a strong incentive to
join the mafia and obtain the protected status.
It can be verified that the
only Nash equilibria 
correspond to outcomes where the central vertex and at most one leaf
are in the Mafia.
}{}

As a different interpretation of the game above, consider a road network with the
vertices representing cities. The maintenance of a road must be
provided by a facility at one of the endpoints. The cost of opening the
facility dominates the operating cost: if city $v$ decides to open one
at cost $c(v)$, it is able to maintain all incident roads. As a
compensation, the cities can try to recollect the opening cost by
asking contributions from the neighboring cities. A city without a
facility has to pay all contributions he is asked to pay. However, if
a city opens a facility, its liability is limited and has to satisfy
demands only up to his opening cost, $c(v)$.

\iftoggle{full}{\smallskip}{}

Our approach can be  extended to the hitting set problem, which is equivalent to the set cover problem.
We are given a hypergraph $G=(V,{\cal E})$, and a cost function
$c:V\rightarrow \R_+$ on the vertices. Our aim is to find a minimum
cost subset $M$ of $V$ intersecting every hyperedge in $\cal E$. This problem is known to be approximable within a factor of $d$,
the maximum size of a hyperedge.
In the corresponding Mafia game, the hyperedges shall be considered as
clubs in need of security. A mafioso can assign ransoms to the clubs he is a
member of, that will be distributed equally to all other members of
the club.

We shall prove that for the vertex cover and hitting set games, the
price of anarchy is 2 and $d$, respectively.
Bar Yehuda and Even gave a simple primal-dual algorithm with this
guarantee in 1981 \cite{Baryehuda81}. No better constant factor
approximation has been given eversince.
Furthermore, assuming the Unique Games Conjecture, Khot and Regev \cite{Khot03} proved that
 the hitting set problem cannot be approximated by any constant factor
 smaller than $d$.

As a further extension, we also investigate the submodular hitting set (or set cover) problem, that has received significant
attention recently.
The goal is to find a hitting set $M$ of a
hypergraph minimizing $C(M)$ for a  submodular set function $C$ on the
ground set. Independently, Koufogiannakis and Young
\cite{Koufogiannakis09greedy} and Iwata and Nagano \cite{Iwata09}
gave $d$-approximation algorithms. Our game approach extends
even to this setting, with the same price of anarchy $d$. This
involves a new agent, the Godfather, who's strategy consists of
setting a budget vector in the submodular base polyhedron of
$C$. Otherwise, the game is essentially the same as  the (linear)
hitting set game.

The main results of the paper can be summarized as follows.
\begin{theorem}
The Mafia games for vertex cover, hitting set and submodular hitting
set always have pure Nash equilibria, and the price of anarchy is 2 for
vertex cover and $d$ for (submodular) hitting set.
\end{theorem}

Recent work of Roughgarden et al. \cite{rough09,rough10,rough11} has shown that the majority of positive results in price of anarchy literature can be reduced to a specific common set of structural assumptions. In contrast,
in our work, we use a novel approach  by exploring connections to the
LP relaxations
of the underlying centralized optimization problems. This connection
raises  interesting questions about the limits of its applicability.
%
%
%
%
%

\iftoggle{full}{\smallskip}{}

\noindent{\bf Convergence and Complexity of Dynamics.}
\iftoggle{full}{
The world of 
 decentralized competition is not immune to the results of
computational complexity. Hence, a low price of anarchy although promising
does not necessarily yield a usable outcome in the means of the game
dynamics,
when agents sequentially have the possibility to change their
strategies for a better one.
The reasons for these inconsistencies fall in one of two possible
categories: either non-convergence of the dynamics to a Nash
equilibrium or too slow convergence. 

Even in a very simple game settings, with a constant number of agents
and strategies and no computational complexity issues, it
could be the case that games exhibit only highly unstable Nash
equilibria. In such settings, numerous learning dynamics, even if they
start off from a state close to a Nash equilibrium, they diverge away
from it fast \cite{daskalakis10,paperics11}.

On the other hand, as the games grow in size, even when there exist
simple decentralized dynamics which provably converge to Nash
equilibria, it is not necessarily the case that this convergence is
achieved within polynomial time. For example, in the case of general
congestion games, although best response dynamics always converge to a
Nash equilibrium, finding any sample equilibrium (even via a
centralized algorithm) has been shown to be PLS-hard
\cite{fabrikant04}, implying that any decentralized dynamic is bound
to fail as well in worst case instances.}{%
A low price of anarchy although promising
does not necessarily yield a usable outcome in the means of the game
dynamics,
when agents sequentially have the possibility to change their
strategies for a better one.
The reasons for these inconsistencies fall in one of two possible
categories: either non-convergence of the dynamics to a Nash
equilibrium  \cite{daskalakis10,paperics11} or too slow convergence \cite{fabrikant04}. 
}

In our covering games, we first show that even in
simple instances, round robin best response dynamics\footnote{These
  are the dynamics where each agent takes turn playing his best
  response in a cyclic ordering according to some fixed permutation.}
may end in a loop.  However, this can be simply fixed by a slight
modification of the payoff.
We introduce a secondary utility, that does not affect the price of
anarchy results, but merely instigates the mafiosi to use more fair
(symmetric) ransoms: $r(u,v)=r(v,u)$. With this secondary objective,
we show that actually a single round of best response dynamics under a
simple selection rule of the next agent results
in a Nash-equilibrium.
This dynamics in fact simulates the
Bar-Yehuda--Even algorithm.
An analogous dynamics is shown in the case of 
hitting set. Moreover, these dynamics can be interpreted in a
distributed manner, enabling several agents to change their
strategies at the same time.

\iftoggle{full}{\smallskip}{}
In our games, the set of strategies is infinite as ransoms can be
arbitrary real numbers. However, if the vertex weights are integers, we
can restrict possible ransoms to be integers as well. All results of
the paper straightforwardly extend to this finite game.

\smallskip

\iftoggle{full}{\subsection{Related work}}{\noindent\textbf{Related work.}}
The basic set cover games in
\cite{buchbinder08}, \cite{Escoffier10} and  \cite{Balcan11}  fall into the class of
congestion games \cite{rosenthal73}. In the models of
\cite{buchbinder08}, \cite{Escoffier10}, in the hitting set terminology,
the agents are the hyperedges that choose a vertex to cover them, and
the cost of the vertex is divided among them according to some rule.
\cite{buchbinder08} investigates the influence of a central authority
that can influence choices by taxes and subsidies in a best response
dynamics;
\cite{Escoffier10} studies different cost sharing rules of the
vertices (``local taxes''). However, none of these methods  achieve
a constant price of anarchy.  The model of
 \cite{Balcan11} can achieve a good equilibrium by assuming  a
central authority
that propagates
information on an optimal solution to a fraction of the agents.
In contrast to
\cite{buchbinder08} and \cite{Balcan11}, our model is defined locally,
without assuming a central authority.

Cardinal and Hoefer \cite{Cardinal10} define a general class of
covering games, including
 the vertex
cover game \cite{Cardinal06}, and also
 the selfish network design game by Anshelevich
et al. \cite{Anshelevich08}.
The game is based on a covering problem given by a linear integer program.
Variables represent resources, and the agents correspond to certain
sets of constraints they have to satisfy. An agent can offer money
for resources needed to satisfy her constraints.
From each variable, the number of units covered by the
total offers of the agents will be purchased and can be used by all
agents simultaneously to satisfy their constraints, regardless to
their actual contributions to the resource.

\iftoggle{full}{
In the vertex cover or hitting set game, the resources are the service points and the
set of constraints belonging to the agents express that every (hyper)edge
owned by them has to be covered. In the model of
 \cite{Anshelevich08},
agent $i$ wants to connect a set of terminals $S_i$ in a graph
$G=(V,E)$ with edge costs $c$. Hence the variables represent the edges
of the graph and the constraints belonging to agent $i$ enforce the
connectivity of $S_i$.

}{}

Our games can be seen as the {\em duals} of these coverings games.
That is, the agents correspond to the variables, and are
responsible for the satisfaction of the constraints containing
them. If a constraint is left unsatisfied, the
participating variables get punished. Also, a variable may require
compensation (ransoms) from other variables participating in the same
constraints. These compensations will correspond to a dual solution in
a Nash equilibrium.
We hope that our approach of studying dual covering games might be
extended to a broader class of problems, with the price of anarchy
matching the integrality gap.

\iftoggle{full}{Our result and the above papers are focused on noncooperative covering games.
A different line of game theoretic study is focused on cost sharing mechanism,
e.g. \cite{Deng99,Devanur05,Immorlica08,Fang07,li05,li10}.}{}

The performance of behavioral dynamics in games  and specifically establishing fast convergence to equilibria of good quality has been
the subject of intensive recent research \cite{Kleinberg09multiplicativeupdates,kleinberg2011load,Shah10}.
The importance of such results that go beyond the analysis of performance
of Nash equilibria has been stressed in
\cite{paperics11} where it has been shown that even in very simple games with constant number of agents and strategies, the performance of simple learning dynamics can be arbitrarily different than (any convex combination of) the payoffs of Nash equilibria.

\iftoggle{full}{
\smallskip

The rest of the paper is organized as follows. Section~\ref{sec:main} defines the
Mafia games for vertex cover, hitting set, and submodular hitting set,
and proves the existence of Nash equilibria and gives price of anarchy
bounds. Section~\ref{sec:dynamics} shows that certain simple dynamics
rapidly converge to Nash equilibrium for vertex cover and for hitting
set. Section~\ref{sec:concl} discusses possible further research
directions.}{%
The rest of this extended abstract is organized as follows. Section~\ref{sec:main} defines the
Mafia games for vertex cover, hitting set, and submodular hitting
set. The proofs of existence of Nash equilibria and the price of
anarchy bound is given only for vertex cover and deferred to the
Appendix for the other two problems.
Section~\ref{sec:dynamics} discusses results on dynamics, and 
Section~\ref{sec:concl}  possible further research
directions.
}

\section{The Mafia games and Price of Anarchy bounds}\label{sec:main}
\subsection{Vertex cover}\label{sec:vc}
Given a graph $G=(V,E)$, let $c:V\to\R^+$ be a cost function on the
vertices.
In the {\em vertex cover problem}, the task is to find a minimum cost
set $M\subseteq V$ containing at least one endpoint of every edge in $E$.
For a vertex $v\in V$, let $N(v)=\{u: uv\in E\}$ denote the set of its
neighbors.

\iftoggle{full}{\smallskip}{}

{\noindent\textbf{Game definition. }}
The {\em Mafia Vertex Cover Game} is a one-shot game  on the agent set $V$.
The basic strategy of an agent is to decide being a civilian or a
mafioso. The set of civilians shall be denoted by $C$,  the set of
mafiosi (Mafia) by $M$.
For civilians, no further decision has to be made, while for
mafiosi, their strategy also contains a ransom vector.
Each mafioso $m\in M$ can demand ransoms from his neighbors totaling
 $c(m)$. The ransom demanded from a neighbor $u\in
N(m)$ is $r(m,u)\ge 0$, with $\sum_{u\in N(v)}r(m,u)=c(m)$.
The strategy profile $\C S=(M,C,r)$ thus consists of the sets of
mafiosi and civilians, and the ransom vectors.

Let us call $c(v)$ the \emph{budget} of an agent $v\in V$, and let
$T>\sum_{v\in V} c(v)$ be a huge constant.
Let $D(v)=\sum_{m\in M} r(m,v)$ be the demand asked from the agent
$v\in V$.

Let us now define the payoffs for a given strategy profile $\C S$.
For a civilian $v\in C$, let $\operatorname{Pen}(v)=T$ if $v$ is
incident to an uncovered edge, that is $C\cap N(v)\neq \emptyset$, and
$\operatorname{Pen}(v)=0$
otherwise. The utility of  $v\in C$ is%
\iftoggle{full}{
\[
U_{\C S}(v) =  - D(v) - \operatorname{Pen}(v).
\]
}{
$U_{\C S}(v) =  - D(v) - \operatorname{Pen}(v)$.
}

If $v\in M$ and the total demand from $v$
is $D(v) > c(v)$ (i.e. $v$ is asked too much),
we call $v$ \emph{protected} and denote the set of protected mafiosi by $P\subseteq M$.
The real amount of money that the protected mafioso $p\in P$ pays to his neighbors is scaled
down to
$\frac{c(p)}{D(p)} r(u,p)$.
Let $F^-(v)=\min\{D(v),c(v)\}$ be the total amount the mafioso $v$ pays
for ransom.
\iftoggle{full}{
Let 
$$
F^+(v) = \sum_{u\in N(v)\setminus P} r(v,u)
+ \sum_{u\in N(v)\cap P} \frac{c(u)}{D(u)} r(v,u)
$$}{

Let
$F^+(v) = \sum_{u\in N(v)\setminus P} r(v,u)
+ \sum_{u\in N(v)\cap P} \frac{c(u)}{D(u)} r(v,u)$
}
denote the income of $v\in M$ from the ransoms.
Then the utility of a mafioso $v\in M$ is defined as
\iftoggle{full}{%
\[U_{\C S}(v) = -c(v)+F^+(v) - F^-(v).\]
}{%
$U_{\C S}(v) = -c(v)+F^+(v) - F^-(v).$

}
This means $v$ has his initial cost $c(v)$ for entering the Mafia,
receives full payment from civilians and unprotected mafiosi,
receives reduced payment from protected mafiosi, and pays the full payment to
his neighboring mafiosi if $v$ is unprotected, or reduced payment if $v$ is protected.

\iftoggle{full}{\smallskip}{}

\noindent\textbf{The existence of pure Nash equilibria.}
Pure Nash equilibria are (deterministic) strategy outcomes such that no agent can improve her payoff by unilaterally changing her strategy.
We will start by establishing that our game always exhibits such states.
The following is the standard linear programming relaxation of vertex
cover along with its dual.
\begin{multicols}{2}{%
\noindent%
\begin{align}
\min& \sum_{v\in V} c(v) x(v) \tag{P-VC}\label{prog:PVC}\\
x(u)+x(v)&\ge 1 \quad \forall uv\in E\notag\\
x&\ge0\notag
\end{align}%
}{%
\begin{align}
\max& \sum_{uv\in E} y(uv) \tag{D-VC}\label{prog:DVC}\\
\sum_{uv\in E} y(uv) &\le c(u) \quad \forall u\in V\notag\\
y&\ge0\notag
\end{align}%
}%
\end{multicols}
For a feasible dual solution $y$ we say
that the vertex $v\in V$ is {\em tight} if $\sum_{uv\in E}y(uv)=c(v)$.
We call the pair $(M,y)$ a {\em complementary pair} if $M$ is a vertex
cover, $y$ is a feasible dual solution, and each $v\in M$ is tight with
respect to $y$. 
\nottoggle{full}{The following well-known claim states that a
complementary solution provides good approximation.}{}

\begin{lemma}\label{claim:vc-approx}
If $(M,y)$ is a complementary pair, then $M$ is a 2-approximate
solution to the vertex cover problem. \nottoggle{full}{\proofbox}{}
\end{lemma}
\iftoggle{full}{
\begin{proof}
The primal objective is at most twice the dual objective, as
\[
\sum_{v\in M}c(v)=\sum_{v\in M}\sum_{u\in N(v)}y(uv)\le 2\sum_{uv\in E}y(uv).
\]
The inequality follows as each edge $uv$ is counted at most twice.
\end{proof}
We shall show that the}{The} simple approximation algorithm by Bar-Yehuda
and Even \cite{Baryehuda81} returns a complementary pair, and
therefore has approximation factor 2.
\nottoggle{full}{We start from $y=0$ and $M=\emptyset$. In each step, we pick an
  arbitrary uncovered edge $uv$, and raise $y(uv)$ until $u$ or $v$
  becomes tight. We add the tight endpoint(s) to $M$ and iterate with a next
  uncovered edge. It is straightforward that the algorithm returns a complementary pair
$(M,y)$.
}{}%
Our next lemma proves that a complementary pair provides a
Nash equilibrium.
\begin{lemma}\label{lemma:vc-dualne}
Let $(M,y)$ be a complementary pair, and consider the strategy profile
where the agents in $M$ form the Mafia and $C=V\setminus M$ are
the civilians. For $u\in M$, define $r(u,v)=y(uv)$ for every $v\in N(u)$.
Then the strategy profile $\C S=(M,C,r)$ is a Nash equilibrium.
\end{lemma}
\begin{proof}
Since $D(v)\le c(v)$ for all players, there are no protected mafiosi.
If $v$ is a civilian, his
payoff is $-D(v)$. He would not get a protected status if he entered
the Mafia as $D(v)\le c(v)$, and thus his payoff would be
$-c(v)+F^+(v)-D(v)\le -D(v)$ by arbitrary choice of ransoms.
If $v$ is a mafioso, he has $F^+(v)=c(v)$ as none of his neighbors is
protected. Thus his utility is $-D(v)$, the maximum he can obtain for
any strategy.\nottoggle{full}{\proofbox}{}
%
%
%
\end{proof}
%
%
\iftoggle{full}{
The existence of a complementary pair is provided by the algorithm of
Bar-Yehuda and Even \cite{Baryehuda81}. 
In
each step of the algorithm we maintain a feasible dual solution, and
$M$ will be the set of tight vertices.
\renewcommand{\labelenumi}{(\arabic{enumi})}
\renewcommand{\labelenumii}{(\arabic{enumi}-\arabic{enumii})}

\begin{enumerate}
\setcounter{enumi}{-1}
\item Set $y(uv):=0$ for each $uv\in E$ and $M=\{v\in V: c(v)=0\}$.
\item While $M$ is not a vertex cover do
\begin{enumerate}
\item Choose an arbitrary edge $uv\in E$ with $u,v\in V-M$.
\item Raise $y(uv)$ until $u$ or $v$ becomes tight.
\item Include the new tight endpoint(s) into $M$.
\end{enumerate}
\item Return $M$.
\end{enumerate}
It is straightforward that the algorithm returns a complementary pair
$(M,y)$. Using Lemma~\ref{lemma:vc-dualne}, we obtain the following.}{
As an immediate consequence, we get the following.}
\begin{theorem}
The Mafia Vertex Cover Game always has a pure Nash equilibrium.\proofbox
\end{theorem}

\iftoggle{full}{\smallskip}{}

\noindent\textbf{The Price of Anarchy. }%
For a strategy profile $\C S$ with $\alpha$ uncovered edges,  the sum ot the utilities
is $-c(M)-2\alpha T$. The Price of Anarchy compares this sum in a Nash
equilibrium at the worst case to the maximum value over all strategy
profiles, that corresponds to a minimum cost vertex cover.

Consider a strategy profile $\C S$ that encodes a Nash equilibrium.
First, observe that Mafia $M$ is a vertex cover\iftoggle{full}{. Indeed, if there
were an uncovered edge $uv\in E$, both $u$ and $v$ would receive the high
penalty $T$, and therefore they would have incentive to join
Mafia.}{ due to the high penalties on uncovered edges.}
We shall prove that the cost $c(M)$ is at most twice the cost of an
optimal vertex cover, consequently, the price of anarchy is at most 2.

\begin{lemma}\label{lemma:vc-noprotected}
Let the strategy profile $\C S=(M,C,r)$ be a Nash equilibrium.
Then there are no protected mafiosi.
\end{lemma}

\begin{proof}
For a contradiction, suppose $P$ is nonempty. First we show there
exists an edge $mp\in E$ such that $m\in M\setminus P$, $p\in P$
and $r(m,p)>0$.
Indeed, if there were no such edges, then $\sum_{p\in P} D(p) \le
\sum_{p\in P} c(p)$
as the ransoms demanded from protected mafiosi are all demanded by 
others $P$. However, by definition  $D(p)>c(p)$ for all $p\in P$, giving
$\sum_{p\in P} D(p) > \sum_{p\in P} c(p)$, a contradiction.

Consider the edge $mp\in E$ as above.
\iftoggle{full}{%
We claim that $m$ could choose a better strategy, and therefore $\C S$
cannot be a Nash equilibrium. If $m$ does not have any civilian
neighbors, that is, $N(m)\subseteq M$, then his utility  would
strictly increase if
decides to become a civilian. Indeed, his income now is $F^+(m)<c(m)$
and he has to pay $F^-(m)=D(m)\le c(m)$. As a civilian, his utility
were $-D(m)$.

Next, assume there exists a $v\in C$, $mv\in E$. Then $m$ may decrease
$r(m,p)$ to 0
and increase $r(m,v)$ by the same amount.
Again, this would be a better strategy for $m$, as $v$ pays the full
amount whereas $p$ payed only a reduced amount.}{%
If $N(m)\subseteq M$, then $M$ could increase his utility if becoming
a civilian, as $F^-(m)=D(m)$ and $F^+(m)<c(m)$, whereas he would
receive $-D(m)$ as a civilian.
If there is a $v\in C$, $mv\in E$, then $m$ could increase his utility
by decreasing $r(m,p)$ to 0
and increasing  $r(m,v)$ by the same amount.
}
\nottoggle{full}{\proofbox}{}\end{proof}

\begin{lemma}\label{lemma:vc-lowdemand}
Suppose the strategy profile $\C S=(M,C,r)$ is a Nash equilibrium and let $v\in C$.
Then $D(v) \le 2c(v)$.
\end{lemma}

\begin{proof}
Suppose the contrary: let $D(v) > 2c(v)$ and thus
 $U_{\C S}(v)<-2c(v)$. If joining Mafia, $v$ receives the
protected status and thus gains utility at least $-2c(v)$ as $F^-(v)=c(v)$.
\nottoggle{full}{\proofbox}{}\end{proof}

\begin{theorem}\label{thm:vc-nash-apx}
The price of anarchy in the Mafia game is 2.
\end{theorem}
\iftoggle{full}{
\begin{proof}
Let $\C S=(M,C,r)$ be a strategy profile in a Nash equilibrium.
Using the convention $r(u,v)=0$ if $u\in C$, let us define $y(uv)=r(u,v)+r(v,u)$
for every edge $uv\in E$.
We show that  $\sum_{u\in V}y(uv)\le
2c(v)$ for every $v\in V$.
Indeed, if $v\in C$, then  $\sum_{u\in V} y(uv) = \sum_{u\in M} r(u,v)  = D(v)
\le 2c(v)$ by Lemma~\ref{lemma:vc-lowdemand}.
If $v\in M$, then $\sum_{u\in V}y(uv)=\sum_{u\in N(v)}r(v,u)+D(v)\le 2c(v)$ by
Lemma~\ref{lemma:vc-noprotected}.
Therefore $\frac12 y$ is a feasible solution to (\ref{prog:DVC})
and
$$
\sum_{uv\in E} \frac12 y(uv) = \frac{1}{2}\sum_{m\in M} \sum_{v\in V} r(m,v)
= \frac12 \sum_{m\in M} c(m).
$$
This verifies that the objective value for $\frac12 y$ is the half of
the cost of the primal feasible vertex cover $M$, proving that $M$ is
a 2-approximate vertex cover.
\end{proof}
}{%
\begin{proof}{\em (Sketch.)} For a strategy profile at Nash
  equilibrium, define $y(u,v)=r(u,v)+r(v,u)$. Using the Lemmas above,
  it follows that $\frac12 y$  is a feasible solution to
  (\ref{prog:DVC}), and moreover the cost of $\frac 12y$ is the half
  of the cost of the primal feasible vertex cover $M$.\proofbox
\end{proof}
}

\subsection{Set cover and hitting set}\label{sec:hs}

In this section, we generalize our approach to the hitting set
problem. Given a hypergraph ${\cal G}=(V,{\cal E})$ and a cost function
$c:V\rightarrow \R_+$, we want to find a minimum cost $M\subseteq V$
intersecting every hyperedge. Let $d=\max\{|S|: S\in {\cal E}\}$.

\iftoggle{full}{
In the {\em set cover problem}, we have a ground set $U$ and a collection of
subsets $\cal S$ of $U$. For a cost function $c:{\cal
  S}\rightarrow \R_+$ we want to find minimum cost collection of subsets whose
union is $U$. This is equivalent to the hitting set problem, where the
ground set is $\cal S$, and to each $u\in U$, there is a corresponding hyperedge
that is the collection of subsets containing $u$.

For simplicity, we define the hitting set game on a $d$-uniform
hypergraph. This can be done without loss of generality. To verify
this, take an arbitrary instance  ${\cal G}=(V,{\cal
  E})$, and  let $T>d\sum_{v\in V}c(v)$. Extend $V$ by $d-1$
new vertices of cost $T$, and for every $S\in {\cal E}$, extend $S$ by any
$d-|S|$ new elements. If there is a $d$-approximate solution
to the modified instance, it cannot contain any of the new
elements. Hence finding a $d$-approximate solution is equivalent
in the original and in the modified instance.
}{
The {\em set cover problem} is well-known to be equivalent to the
hitting set problem. Also, without loss of generality we may define
the hitting set game on a $d$-uniform hypergraph. The general case can
be easily reduced to it by adding at most $d-1$ new dummy elements of
high cost (see Appendix A).
}{}

\iftoggle{full}{\smallskip}{}

\noindent\textbf{Game definition. }We define the {\em Mafia Hitting Set Game} on a $d$-uniform
hypergraph $\C G=(V,\C E)$. The set of agents is $V$, with $v\in V$
having a {\em budget} $c(v)$.  We shall call the
hyperedges {\em clubs}. For an agent $v\in V$, let $\C N(v)\subseteq \C E$ denote the set
of clubs containing $v$.
The agents again choose from the strategy of being a civilian
or being a mafioso, denoting their sets by $C$ and $M$, respectively.
The strategies of the mafioso $m$ incorporates the ransoms $r(m,S)$
for the clubs $S$ containing $m$, with $\sum_{S\in \C N(v)} r(m,S)
= c(m)$.

We define the payoffs for the strategy profile $\C S=(M,C,r)$
similarly to the vertex cover case. For a civilian $v\in C$,
$\operatorname{Pen}(v)=T$ for a large constant $T$ if $v$ participates
in a club containing no mafiosi, and 0 otherwise.

In each club $S$, the ransom $r(m,S)$ of a mafioso $m\in S\cap M$
has to be payed by all other members at equal rate, that is,
everyone pays  $\frac{r(m,S)}{(d-1)}$ to $m$. The
demand from an agent is the total amount he has to pay in
all clubs he is a member of, that is,
\iftoggle{full}{\[
D(v)=\frac1{d-1}\sum_{S\in \C N(v)} \sum_{m\in M\cap S}r(m,S).
\]}{%
\[
D(v)=\frac1{d-1}\sum_{S\in \C N(v)} \sum_{m\in M\cap S}r(m,S).
\]
}
The utility of a civilian  $v\in C$ is defined as
$U_{\C S}(v) =  - D(v) - \operatorname{Pen}(v)$.

A mafioso $v$ receives the protected status if $D(v)>c(v)$.
The set of
protected mafiosi is denoted by $P$, and they pay proportionally reduced ransoms.
Let $F^-(v)=\min\{D(v),c(v)\}$ be the total amount $v$ pays.
The income is defined by
\iftoggle{full}{
$$
F^+(v) = \sum_{S\in \C N(v)}\frac{r(v,S)}{d-1}\left(| S\setminus (P\cup\{v\})|+
\sum_{u\in (S\cap P)\setminus\{v\}}\frac{c(u)}{D(u)}\right).
$$}{%
\[
F^+(v) = \sum_{S\in \C N(v)}\frac{r(v,S)}{d-1}\left(| S\setminus (P\cup\{v\})|+
\sum_{u\in (S\cap P)\setminus\{v\}}\frac{c(u)}{D(u)}\right).\]
}
The  utility of a mafioso $v\in M$ is then
$U_{\C S}(v) =-c(v)+ F^+(v) - F^-(v)$.

\nottoggle{full}{%
Analogously to vertex cover, we show the following.
\begin{theorem}\label{thm:hitting-main-icalp}
There exists pure Nash equilibria in the  Mafia Hitting Set Game, and
the Price of Anarchy is at most $d$. The output of the
Bar-Yehuda--Even algorithm always gives a Nash equilbrium.\proofbox
\end{theorem}
The proof is deferred to Appendix A. It is similar to the case of 
vertex cover, with the proof of the
analogue of
Lemma~\ref{lemma:vc-lowdemand} being the difficult part.}{
\smallskip

\noindent\textbf{The existence of pure Nash equilibria. }
The standard LP-relaxation extends the formulations (\ref{prog:PVC})
and (\ref{prog:DVC}).
\begin{multicols}{2}{%
\noindent%
\begin{align}%
\min& \sum_{v\in V} c(v) x(v) \tag{P-HS}\label{prog:PHS}\\%
\sum_{u\in S}x(u)&\ge 1 \quad \forall S\in \C E\notag\\
x&\ge0\notag%
\end{align}
}{%
\begin{align}%
\max& \sum_{S\in\C E} y(S) \tag{D-HS}\label{prog:DHS}\\
\sum_{S\in \C N(u)} y(S) &\le c(u) \quad \forall u\in V\notag\\
y&\ge0\notag%
\end{align}%
}%
\end{multicols}

Again, for a feasible dual solution $y$, $v\in V$ is called tight if
the corresponding inequality in (\ref{prog:DHS}) holds with equality.
A pair $(M,y)$ of a hitting set $M$ and a feasible dual $y$ is called
a {\em complementary} pair if the dual inequality corresponding to any
$v\in M$ is tight.
The following simple claim generalizes Lemma~\ref{claim:vc-approx}.
\begin{lemma}\label{claim:hs-approx}
If $(M,y)$ is a complementary pair, then $M$ is a $d$-approximate
solution to the hitting set problem.\proofbox
\end{lemma}
The algorithm of Bar-Yehuda and Even \cite{Baryehuda81}, outlined in
Section~\ref{sec:vc} naturally extends to the hitting set
problem, and delivers a complementary pair.

\begin{lemma}\label{lemma:hs-dualne}
Let us define strategies in the Mafia Hitting Set Game based on a
complementary pair $(M,y)$ as follows. Let agents in $M$ be the Mafia and $V\setminus M$ be
the civilians. For each $v\in M$, define $r(v,S)=y(S)$ for every $S\in\C E$
containing $v$.
Then the strategy profile $\C S=(M,C,r)$ is a Nash equilibrium.
\end{lemma}

\begin{proof}
For each $v\in V$,  $D(v)\le \frac{1}{d-1}\sum_{S\in \C
  N(v)}y(S)|(S\cap M)\setminus\{v\}|\le c(v)$ and therefore there are no protected mafiosi.
The proof that nobody has an incentive to change his strategy is the
same as for Lemma~\ref{lemma:vc-dualne}.
\nottoggle{full}{\proofbox}{}\end{proof}

As the algorithm of Bar-Yehuda and Even \cite{Baryehuda81} provides a
complementary pair, this immediately yields the following.
\begin{theorem}
The Mafia Hitting Set Game always has a pure Nash equilibrium.
\end{theorem}

\iftoggle{full}{\smallskip}{}

\noindent\textbf{The Price of Anarchy. }
\begin{lemma}\label{lemma:hs-noprotected}
Let the strategy profile $\C S=(M,C,r)$ be a Nash equilibrium.
Then there are no protected mafiosi.
\end{lemma}

\begin{proof}
The proof follows the same lines as for
Lemma~\ref{lemma:vc-noprotected}. For a contradiction, assume $P\neq \emptyset$.
First, it is easy to show that there exists an unprotected $m\in M\setminus P$ and $S\in\C E$, $S\cap P\neq\emptyset$,
such that $r(m,S)>0$ by comparing the total in-demand and
out-demand of protected mafiosi.
For such $m$, if there exists no set $S'\in\C E$ with
 $S'\cap M = \{m\}$, then he could increase his utility by leaving the
 Mafia.
Otherwise, he could increase his utility by decreasing $r(m,S)$ and
increasing $r(m,S')$.
\nottoggle{full}{\proofbox}{}\end{proof}

The following lemma is the analogue of Lemma~\ref{lemma:vc-lowdemand},
yet the proof is more complicated.
\begin{lemma}\label{lemma:hs-lowdemand}
Let the strategy profile $\C S=(M,C,r)$ be a Nash equilibrium and let $v\in C$.
Then $D(v)\le \frac{d}{d-1} c(v)$.
\end{lemma}

\begin{proof}
Suppose the contrary, let there be a $v\in C$ such that $D(v)>\frac{d}{d-1}c(v)$.
His current utility is $U_{\C S}(v)=-D(v)$.

We show that $v$ could join Mafia and set ransoms that provide him
a strictly larger utility. If $F_{\C S'}^+(v)$ is the income for such a
strategy profile $\C S'$, then $U_{\C S'}(v)=F_{\C S'}^+(v)-2c(v)$,
as he would obtain the protected status. To get $U_{\C S'}(v)>U_{\C
  S}(v)$, we need to ensure $F_{\C S'}^+(v)>2c(v)-D(v)$.
As $D(v)>\frac{d}{d-1} c(v)$ is assumed, it suffices to give an $\C S'$ with
\begin{equation}\label{eq:income-bound}
F_{\C S'}^+(v)\ge \frac{d-2}{d-1}c(v).
\end{equation}
We define the ransoms $r'(v,S)$ by ``stealing'' the strategies of the other
mafiosi.
That is, for each club $S\in \C N(v)$, $\frac{1}{d-1}\sum_{m\in M\cap
  S}r(m,S)$ is the total ransom $v$ has to pay to the members of this
club. We define $r'(v,S)$ proportionally to this amount:
\[
r'(v,S):=\frac{c(v)}{(d-1)D(v)}\sum_{m\in M\cap   S}r(m,S).
\]
By Lemma~\ref{lemma:hs-noprotected}, we know that there are
no protected mafiosi in the Nash equilibrium $\C S$.
We show that after $v$ enters Mafia, even if some of the old mafiosi
become protected, they are only slightly overcharged.
More precisely, we shall show that
\begin{equation}\label{eq:overcharge}
D'(t)\le \frac{d}{d-1} c(t)\ \ \forall t\in M.
\end{equation}
From this bound, (\ref{eq:income-bound}) immediately follows.
Indeed, everybody will pay at least $\frac{d}{d-1}$ fraction of
  the demands, and therefore $F_{\C S'}^+(v)\ge \frac{d-1}{d}c(v)\ge \frac{d-2}{d-1}c(v)$.

It is left to prove (\ref{eq:overcharge}). The demand of $v$ from some
$t\in M$ can be bounded as follows:
\begin{align*}
\frac{1}{d-1}
\sum_{S\in \C N(v)\cap \C N(t)} r'(v,S) = \frac{c(v)}{(d-1)^2D(v)}  \sum_{S\in
  \C N(v)\cap\C N(t)}
\left( r(t,S) + \sum_{t'\in (S\cap M)\setminus\{t\}} r(t',S) \right)\\
  < \frac{1}{d(d-1)} (c(t)+(d-1)D(t)) \le \frac{1}{d(d-1)} d\cdot c(t) = \frac{c(t)}{d-1}.
\end{align*}
Here we used that $D(t)\le c(t)$ as $t$ was not
protected in $\C S$. Using this fact once more, we get
\[
D'(t) \le D(t) + \frac{c(t)}{d-1} \le \frac{d}{d-1} c(t).
\]%
\nottoggle{full}{\proofbox}{}\end{proof}

\begin{theorem}
The price of anarchy for the Mafia Hitting Set Game is $d$.
\end{theorem}

\begin{proof}
Let $\C S=(M,C,r)$ be a strategy profile in a Nash equilibrium.
Then $M$ is a hitting set, as if there was an uncovered club, all members
would be unhappy due to the term $\operatorname{Pen}(v)$. We show that the
cost of $M$ is within a factor $d$ from the optimum.
Let us set $y(S)=\sum_{v\in M\cap S} r(v,S)$ for each $S\in\C E$.
Lemmas~\ref{lemma:hs-noprotected} and \ref{lemma:hs-lowdemand} easily
imply
$\sum_{S\in\C S: v\in S} y(S)\le d\cdot c(v)$ for every $v\in V$, and
thus $\frac 1d y$ is a feasible dual solution to (\ref{prog:DHS}).
Then
$$
\sum_{S\in\C S} \frac{1}{d} y(S) = \sum_{s\in \C S} \frac{1}{d} \sum_{m\in M\cap S}r(m,S)
=\sum_{m\in M} \frac{1}{d} \sum_{S\in \C N(m)} r(m,S) = \frac 1d \sum_{m\in M} c(m),
$$
showing that $M$ is a $d$-approximate solution to (\ref{prog:PHS}).
\end{proof}
}
\subsection{Submodular hitting set}\label{sec:submod}
In the submodular hitting set problem, we are given a hypergraph
$G=(V,\C E)$ with a submodular set function $C:2^V\rightarrow
\R_+$, that is, $C(\emptyset)=0$, and
\iftoggle{full}{
\[
C(X)+C(Y)\ge C(X\cap Y)+C(X\cup Y)\ \ \ \forall X,Y\subseteq V.
\]
}{%
$C(X)+C(Y)\ge C(X\cap Y)+C(X\cup Y)$ for any sets $X,Y\subseteq V$.
}
We shall assume also that $C$ is monotone, that is, $C(X)\le C(Y)$ if
$X\subseteq Y$. Our aim is to find a hitting set $M$ minimizing
$C(M)$.

Koufogiannakis and Young \cite{Koufogiannakis09greedy}, and Iwata and
Nagano \cite{Iwata09} obtained $d$-approximation algorithms for this
problem, where $d$ is the maximum size of a hyperedge.
\iftoggle{full}{%
We shall present the primal-dual algorithm in \cite{Iwata09},  a
natural extension of the Bar-Yehuda--Even algorithm.}{
The primal-dual algorithm in  \cite{Iwata09} is  a
natural extension of the Bar-Yehuda--Even algorithm (see Appendix B).
}

\iftoggle{full}{
For a submodular function $C$, it is natural to define the following
two polyhedra. The
{\em submodular polyhedron} is
\[
P(C)=\{z\in \R^V: z\ge 0, z(Z)\le C(Z)\ \ \forall Z\subseteq V\},
\]
and the {\em submodular base polyhedron} is}{
A notion needed to define our game is the {\em submodular base polyhedron}:
}
\[
B(C)=\{z\in \R^V: z\ge 0, z(Z)\le C(Z)\ \ \forall Z\subsetneq V, z(V)=C(V)\}.
\]
\iftoggle{full}{Given a vector $z\in P(C)$, the set $Z$ is {\em tight} with respect to
$z$ if $z(Z)=C(Z)$. An elementary consequence of submodularity is that
for every $z\in P(C)$, there exists a unique maximal tight set.
Note that $B(C)\subseteq P(C)$ and $z\in P(C)$ is in $B(C)$ if and
only if $V$ is tight.

In the LP relaxation, we assign a primal variable $\xi(Z)$ to every
subset $Z\subseteq V$. In an integer solution, $\xi(Z)= 1$ if $Z$ is the
chosen hitting set and 0 otherwise.
\begin{multicols}{2}{%
\noindent%
\begin{align}
\min& \sum_{Z\subseteq V} C(Z) \xi(Z)\notag\\
\sum_{Z\in \C N(u)}\xi(Z)&=x(u)&\forall u\in V\notag\\
\sum_{u\in S}x(u)&\ge 1&\forall S\in \C E\tag{P-SHS}\label{prog:PSHS}\\
\xi&\ge0\notag
\end{align}%
}{%
\begin{align}
\max& \sum_{S\in\C E} y(S)\notag\\
\sum_{S\in \C N(u)} y(S) &= z(u)&\forall u\in V\notag\\
z&\in P(C)\tag{D-SHS}\label{prog:DSHS}\\
y&\ge0\notag
\end{align}%
}%
\end{multicols}

Note that in the dual program, $y$ uniquely defines $z$. Therefore we
will say that $y$ is a feasible dual solution if the corresponding $z$
is in $P(C)$. For the special case of the (linear)
hitting set problem, where $C(Z)=\sum_{v\in Z}c(z)$ for some
$c:V\rightarrow\R_+$, this is equivalent to $y$ satisfying (\ref{prog:DHS}).

Accordingly, we say that a set $Z$ is tight for a
feasible dual $y$ if $z(Z)=P(C)$.
For a hitting set $M$ and a feasible dual solution $y$, we say that
$(M,y)$ is a {\em complementary pair} if $M$ is tight for $y$. The
following is the
generalization of Lemmas~\ref{claim:vc-approx} and \ref{claim:hs-approx}.
\begin{lemma}\label{claim:shs-approx}
If $(M,y)$ is a complementary pair, then $M$ is a $d$-approximate
solution of the hitting set problem.
\end{lemma}
\begin{proof}
The primal objective is at most $d$ times the dual objective, as
\[
C(M)=\sum_{v\in M}z(m)=\sum_{v\in M}\sum_{S\in \C N(v)}y(S)\le d\sum_{S}y(S).
\]
The inequality follows as each $S$ is counted $|S|\le d$ times.
\nottoggle{full}{\proofbox}{}\end{proof}

The algorithm by Iwata and Nagano \cite{Iwata09} is as follows.
\begin{enumerate}
\setcounter{enumi}{-1}
\item Set $y(S):=0$ for each $S\in \C E$, $z(v):=0$ for $v\in V$,  and let $M$ be the unique maximal
  set with $C(M)=0$.
\item While $M$ is not a hitting set do
\begin{enumerate}
\item Choose an arbitrary hyperedge $S\in \cal E$, $S\cap M=\emptyset$.
\item\label{step:epsilon} Compute $\varepsilon=\max\{\lambda:z+\lambda\chi_Z\in P(C)\}.$
\item Increase $y(S)$ and every $z(v)$ for $v\in Z$ by $\varepsilon$.
\item Replace $M$ by the new unique maximal tight set.
\end{enumerate}
\item Return $M$.
\end{enumerate}

In step (1-2), $\chi_Z$ is the characteristic function of $Z$. This
step can be performed in the same running time as a submodular
function minimization (see \cite{fleischer03}). Note also that $M$ will always intersect $S$ in
step (1-4) and therefore will be strictly extended. It is immediate that it
returns a complementary pair $(M,y)$ and thus
Lemma~\ref{claim:shs-approx} proves $d$-approximation.

\iftoggle{full}{\smallskip}{}
}
\noindent\textbf{Game definition. }%
\iftoggle{full}{The vector $z$ in (\ref{prog:DSHS}) plays an analogous role to the
budgets $c$ in the (linear) Mafia Hitting Set Game.
We introduce a new agent, the {\em Godfather} to set the budgets of
the agents.

}{}
The {\em Submodular Mafia Hitting Set Game} is defined on
 a hypergraph $\C G=(V,\C E)$ and a monotone
 submodular set function $C:2^V\rightarrow \R_+$. There are $|V|+1$
agents, one for each vertex and a special agent $g$, called the {\em
  Godfather}.

The strategy of the Godfather
is to return a budget vector $\tilde c\in B(C)$.
The basic strategy of an agent $v\in V$ is to
decide being a civilian or being a mafioso.
The strategy of a mafioso $m$ further incorporates
normalized ransoms $r_0(m,S)\ge 0$ for clubs $S\in \C N(m)$ with
$\sum_{S\in\C N(m)}r_0(m,S)=1$, that is, $r_0(m,S)$ expresses the
fraction of the budget of $m$ he is willing to charge on $S$.

The sets of civilians and mafiosi will again be denoted
by $C$ and $M$, respectively. Hence a strategy profile is given as $\C
S=(M,C,\tilde c,r_0)$.
The actual ransoms will be $r(m,S)=r_0(m,S)\cdot\tilde c(m)$.

The utility of the Godfather is the total budget of the Mafia: $U_{\C S}(g)=C(M)$.
The utility of the vertex agents is defined the same way as for the
linear Mafia Hitting Set Game in  Section~\ref{sec:hs}, with replacing $c(v)$ by $\tilde c(v)$ everywhere.

For linear cost functions, we have $C(Z)=\sum_{v\in Z}c(z)$. Then
the only vector in $B(C)$ is $c$, hence the Godfather has
only one strategy to choose. Therefore we obtain the same game as
described in Section~\ref{sec:hs}.
\nottoggle{full}{
Our result can be summarized as follows.
\begin{theorem}\label{thm:submod-main-icalp}
There exists pure Nash equilibria in the  Submodular Mafia Hitting Set Game, and
the Price of Anarchy is at most $d$. The output of the primal-dual
algorithm by Iwata and Nagano \cite{Iwata09} always gives a Nash equilbrium.
\end{theorem}
The proof (see Appendix B) reduces to the linear Mafia Hitting Set Game, exploiting the
fact the whenever the Godfather has no incentive to change his
strategy, from the perspective of the vertex players it is identical
to a linear game with fixed budgets $\tilde c$.
}{

\smallskip

\noindent\textbf{Existence of a Nash equilibrium and bounding the Price of
  Anarchy. }
As for vertex cover and hitting set, we show that a complementary
solution $(M,y)$ to (\ref{prog:PSHS}) and (\ref{prog:DSHS}) provides a
solution in Nash equilibrium. Let $z(u)=\sum_{S\in \C N(u)} y(S)$.
Note that $z\in P(C)$ and $M$ is tight for $z$.
Let us raise the $z(v)$ values for $v\in C$ arbitrarily in order to get
a vector in the base polyhedron $B(C)$. Let $\tilde c$ denote such a vector.

\begin{lemma}\label{lemma:shs-dualne}
Let us define strategies in the Mafia Hitting Set Game based on a
complementary pair $(M,y)$ with $M$ being the Mafia and $V\setminus M$ the civilians. Let the
Godfather assign the budget vector $\tilde c$ as defined above.
For $u\in M$ and $S\in \C N(u)$, define $r_0(u,S)=y(S)/\tilde c(u)$.
Then the strategy profile $\C S=(M,C,\tilde c,r_0)$ is a Nash equilibrium.
\end{lemma}
\begin{proof}
The Godfather has no incentive to change as by $\tilde c(M)=C(M)$, he
already receives the maximum possible utility for the given $M$.
By the definition, $r(u,S)=y(S)$ for each $v\in V$,  hence $D(v)\le
\frac{1}{d-1}\sum_{S\in \C N(v)}y(S)|(S\cap M)\setminus\{v\}|\le
\tilde c(v)$ and therefore there are no protected mafiosi.
The proof that nobody has an incentive to change his strategy is the
same as for Lemma~\ref{lemma:vc-dualne}.
\nottoggle{full}{\proofbox}{}\end{proof}

\begin{theorem}
The price of anarchy for the Submodular Mafia Hitting Set Game is $d$.
\end{theorem}
\begin{proof}
Consider a strategy profile $\C S=(M,C,\tilde c,r_0)$ in a Nash equilibrium.
We can repeat the entire argument of Section~\ref{sec:hs} to show
that there are no protected mafiosi and that every civilian is
demanded at most $\frac{d}{d-1} \tilde c$. This is since if the Godfather
does not change his strategy, the game is identical to the linear game
with fixed budgets $\tilde c$ from the perspective of the vertex agents.

Finally we define $y(S)=\sum_{m\in M\cap S}r(m,S)$, and then $\frac 1d y$
is a feasible dual solution for (\ref{prog:DSHS}).
Observe that
$\tilde c(M)=C(M)$, as otherwise the Godfather would have incentive to
decrease the budgets of certain civilians and increase for certain
mafiosi.
Consequently,
$C(M)=\tilde c(M)=\sum_{S\in \C E}y(S)$, showing that $M$ is a
$d$-approximate solution for (\ref{prog:PSHS}).
\end{proof}
}
\section{Convergence to Nash equilibrium}\label{sec:dynamics}
\iftoggle{full}{
In this section, we investigate if the Mafia Games defined in the
previous section converge under certain best response dynamics.
We first show that already in the Mafia Vertex Cover Game, a round
robin best response dynamics may run into a loop.

Motivated by this
example, we modify the utilities by adding a secondary payoff, that
instigates the mafiosi to use symmetric ransoms: $r(u,v)=r(v,u)$.
With this secondary objective, we show that a single round of best
response dynamics under a simple selection rule results in a
Nash-equilibrium.
This dynamics simulates the
Bar-Yehuda--Even algorithm.
An analogous result is then proved for hitting set.
Finally, we discuss possible extensions for the submodular case.

\subsection{Vertex cover}
Let us now show an example where a round robin dynamics does not
necessarily converge.
Consider a star on the vertices $v_1,v_2,v_3,v_4$ and the central vertex
$z$.
Assume we are playing round robin in the order $z,v_1,v_2,v_3,v_4$.
Let $c(v_i)=1$ for $i=1,2,3$ and $c(z)=2$, and let us start with the
strategy profile where  $M=\{z\}$, $r(z,v_1)=2$.
Assume that whenever $z$ can change his strategy to get a higher
utility, he always chooses to demand his entire budget 2 from one of
the civilians among $v_1,v_2,v_3,v_4$ (this is always a best response).

We claim that this will always be possible as $z$ always stays in the
Mafia, and at most 3 vertices among
$v_1,v_2,v_3,v_4$ will be in the Mafia at the same time.
Indeed, a civilian will enter only if being ransomed by
$z$. If $v_i$ is in the Mafia then his only option is
setting $r(v_i,z)=1$, and thus if $z$ has at least 3 neighbors in the
Mafia, he becomes protected and thus all his neighbors he is not
actually ransoming will have an
incentive to leave.

The dynamics never reaches a Nash equilibrium, as if $z$ is ransoming a
mafioso $v_i$,  he has incentive to change to ransoming a civilian as
$v_i$ is protected. On the other hand,  if $z$ ransoms a civilian $v_i$, $v_i$
has an incentive to join the
Mafia to obtain the protected status.

\iftoggle{full}{\smallskip}{}

If we could incentivize $z$ to change his
strategy less drastically and ransom the other players by at most 1, 
we could rapidly reach a Nash-equilibrium.
To enforce such a behavior, we introduce a secondary utility
function.}{
In this section, we investigate if the Mafia Vertex Cover Game from
the best response dynamics perspective. An example (described in
Appendix C) shows
that a best response dynamics can run into a loop.
 The problem is due to assymetric ransoms between mafiosi.
We introduce the secondary utilities motivated by this phenomenon.%

}
For a strategy profile $\C S=(M,C,r)$,  $U_{\C S}(v)$ is the
utility as defined in Section~\ref{sec:vc}.
Let us define $\tilde U_{\C S}(v)=0$ if $v\in C$ and
\iftoggle{full}{
\[
\tilde U_{\C S}(v)=-\sum_{uv\in E,u\in M} |r(u,v)-r(v,u)|
\]}{%
$\tilde U_{\C S}(v)=-\sum_{uv\in E,u\in M} |r(u,v)-r(v,u)|$%
}
if $u\in M$. The total utility is then $(U_{\C S}(v),\tilde U_{\C
  S}(v))$ in the lexicographic ordering: the agents' main objective is
to maximize $U_{\C S}(v)$, and if that is the same for two
outcomes, they choose the one maximizing $\tilde U_{\C
  S}(v)$. \iftoggle{full}{In the above example, the dynamics would reach an
equilibrium in the second round, with $r(z,v_i)\le 1$ for all $i$.}{}

 $\tilde U_{\C S}(v)\le 0$ and equality holds  if
 $r(u,v)=r(v,u)$ for every $uv\in M$, $u,v\in M$. Therefore all
 results in Section~\ref{sec:vc} remain valid: in
 Lemma~\ref{lemma:vc-dualne} we define a strategy profile where
 $\tilde U_{\C S}(v)= 0$ for all agents, hence it also gives a Nash equilibrium for
 the extended definition of utilities. The secondary utility term
 $\tilde U$ does not affect the proofs in Section~\ref{sec:vc}.

Consider now the following simple dynamics: {\em Start from the strategy profile
where all agents are civilians. In each step, take an agent who is
incident to uncovered edge and subject to this, minimizes
$c(v)-D(v)$, and give him the opportunity to change his
strategy.}

\begin{theorem}\label{thm:vc-dynamics}
After each agent changing his strategy at most once, we obtain a
strategy profile in Nash equilibrium. \nottoggle{full}{\proofbox}{}
\end{theorem}
\nottoggle{full}{The proof is deferred to Appendix C. The resulting
  dynamics is}{
\begin{proof}
 By induction, we shall prove that in every
step, $c(v)\ge D(v)$ and $\tilde U_{\cal S}(v)=0$ for all $v\in V$.
Consider the next move, when a player $v$ incident to some uncovered
edges minimizing $c(v)-D(v)$ moves.
He obviously has to enter the Mafia, and can achieve a maximal
(primary and secondary) utility if he sets $r(v,u)=r(u,v)$ for any
$u\in M\cap N(v)$, and distributes the rest of his ransoms arbitrarily to his
civilian neighbors. Note that this can always be done because
$c(v)\ge D(v)$. Also, note that the total ransom $v$ will demand from
other civilians is $c(v)- D(v)$. By the extremal choice of $v$, it
follows that none of his civilian neighbors $z$ will violate $c(z)\ge
D(z)$. This also remains true if $z\in M$, as $D(z)$ is at most
the total ransom $z$ demands due to the symmetry of the ransoms.

Hence the induction hypothesis is maintained by an arbitrary best
response of $v$. A mafioso who is not protected and has secondary
objective 0 has no incentive to change his strategy. Also, a civilian
$v$ with $c(v)\ge D(v)$ has no incentive to join the Mafia if there
are no uncovered edges incident to $v$. Consequently, the game ends
after all uncovered edges are gone, and once an agent joins to Mafia, he
would not change his strategy anymore.
\end{proof}

Observe that the dynamics is} closely related to the
Bar-Yehuda--Even algorithm: if the next agent always ransoms only one
of its civilian neighbors, then it corresponds to a possible performance
of the algorithm.

\iftoggle{full}{\smallskip}{}
The above dynamics can be naturally interpreted in a distributive
manner. In the proof of Theorem~\ref{thm:vc-dynamics}, we only use that the vertex $v$ changing his
strategy {\em is a local minimizer} of $c(v)-D(v)$.
The simultaneous move of two agents $u$ and $v$ could interfere only
if $uv\in E$ or
they have a neighbor $t$ in common. In this case, $c(t)<D(t)$ could
result if both $u$ and $v$ start ransoming $t$ simultaneously.

We assume that the agents have a hierarchical ordering
$\prec$: $u\prec v$ expresses that $v$ is more powerful than
$u$. We call an agent $v$ {\em a local minimizer} if $v\in C$, $v$ is incident to
some uncovered edges, and
$c(v)-D(v)\le c(u)-D(u)$ whenever $u\in C$, $uv\in E$. A local
minimizer $v$ is then called {\em eligible} if $u\prec v$ for all local
minimizers $u$ whose distance from $v$ is at most 2.

We start from $C=V$. In each iteration of the dynamics, we let all eligible agents change their
strategy to a best response simultaneously. As in the proof of Theorem~\ref{thm:vc-dynamics},
$c(v)-D(v)\ge 0$ is maintained for all $v\in V$, and thus the
dynamics terminates after each agent changes his strategy at most once.

\iftoggle{full}{
There are multiple distributed algorithms in the literature for vertex
cover, e.g. \cite{khuller94,grandoni05,Koufogiannakis09}. The
distributed algorithm by Koufogiannakis and Young
\cite{Koufogiannakis09} computes in $O(\log n)$ rounds a
2-approximation in expectation with high probability. In contrast, }{
In contrast to efficient distributed algorithms for vertex cover in
the literature (e.g. \cite{Koufogiannakis09}),}
we cannot give
good bounds on the number of iterations of our distributed
dynamics. For example, if the graph is a path $v_1\ldots v_n$,
 and the budgets are $c(v_i)=i$, then only agent $i$ will move in
 step $i$. Yet we believe that our dynamics could be practically efficient.

\iftoggle{full}{

\subsection{Hitting set}
The natural generalization of the secondary objective for hitting set
is as follows. For a club $S\in \C E$, let $var(S)$ denote the maximum
difference between ransoms on this edge. That is, $var(S)=0$ if
$|S\cap M|\le 1$ and $var(S)=\max_{v\in S\cap M}r(v,S)-\min_{v\in
  S\cap M}r(v,S)$ otherwise.
For a strategy profile $\C S=(M,C,r)$, let $\tilde U_{\C S}(v)=0$ if
$v\in C$ and
\[
\tilde U_{\C S}(v)=-\sum_{v\in \C N(v)}var(S)
\]
if $v\in M$. The utility of an agent is then $(U_{\C S}(v),\tilde
U_{\C S}(v))$, under lexicographic ordering.

A natural expectation would be to prove rapid convergence as for
vertex cover, if always the agent minimizing $c(v)-D(v)$  is allowed
to play. However, the Bar-Yehuda--Even algorithm does not seem to be
modeled by this dynamics.
Instead, we define a slightly different extremal choice of the next
agent. Let
\[D^*(v)=\sum_{S\in \C N(v)} \max_{m\in S\cap
  M}r(m,S),\]
 that is, for each club $S$ we consider the largest ransom
demanded in this club. Note that $D(v)\le D^*(v)$.
Let us consider the following dynamics. {\em We start from the strategy
profile where everyone is civilian, and we always let a civilian play
next who is contained in an uncovered club. Among them, we let
the one play who minimizes $c(v)-D^*(v)$.}

\begin{theorem}\label{thm:hs-dynamics}
After each agent changing his strategy at most once, we obtain a
strategy profile in Nash equilibrium.
\end{theorem}
\begin{proof}
We prove by induction, that in every step,
 $c(v)\ge D^*(v)$ and $\tilde U_{\cal S}(v)=0$ for all $v\in V$.
Note that this implies that there are no protected mafiosi.
If $M$ is not a hitting set, we let a $v$ minimizing $c(v)-D^*(v)$ play.
$\tilde U_{\cal S}(m)=0$ for all $m\in M$ means that for every club $S$,
$r(m,S)$ is equal for every $m\in M\cap S$; let $r_S$ denote this
common value.
As for the vertex cover case, the best responses of $v$ are to set
$r(v,S)=r_S$ whenever $S$ was already covered by the Mafia, and
to distribute the remaining ransoms arbitrarily on the hyperedges
covered only by $v$.

As $D^*(v)=\sum_{S\in \C N(v)}r_S$, the remaining
amount $v$ distributes is exactly $c(v)-D^*(v)$. Then by the choice of
$v$, $c(z)\ge D^*(z)$ shall be maintained for every civilian $z$, and
also for other mafiosi (note that if $z\in M$, then $D^*(z)$ does not change).
It can be seen analogously as for vertex cover, that we have a Nash
equilibrium if there are no more
uncovered clubs.
\nottoggle{full}{\proofbox}{}\end{proof}

Similarly to the vertex cover case, this dynamics essentially simulates the
Bar-Yehuda--Even algorithm. Also, an analogous distributed interpretation can be
given.

\subsection{Submodular hitting set}
One would expect that the Submodular Mafia Hitting Set Game also
converges under some dynamics that simulates the primal-dual algorithm
by Iwata and Nagano \cite{Iwata09}. However, if the Godfather does not
have a secondary utility, the following example shows that it can
run into a loop even in very simple instances.

Let $V=\{a,b\}$, $C(\{a\})=C(\{b\})=C(\{a,b\})=1$, and let $g$ be the
Godfather. Let $\C E=\{\{a,b\}\}$; for simplicity, we use the notation
of vertex cover, e.g. $r_0(a,b)$ denotes $r_0(a,\{a,b\})$.
Let us start from the strategy profile $C=V$, $\tilde c(a)=1$, $\tilde
c(b)=0$, and play a round robin in the order $a,b,g$.

First, $a$ enters Mafia and sets $r_0(a,b)=1$. Then $b$ also enters to
receive the protected status and sets $r_0(b,a)=1$.
$g$ has no incentive to move as $\tilde
c(M)=1$ is already maximal. In the next round, $a$ is happier if he
leaves Mafia; $b$ has no incentive to change, however $g$
modifies to $\tilde c(a)=0$ and $\tilde c(b)=1$. This will lead to a
loop: $a$ enters again in next round, $b$ leaves, $\tilde c$ is
changed again, etc.

The above behavior can be avoided by introducing a secondary
utility for $g$: let $\tilde U_{\C S}(g)=\sum_{v\in M}F^+(v)$, that
is, the sum of the actual incomes of the mafiosi. Note that $\tilde
U_{\C S}(g)\le \tilde c(M)$ and equality holds if and only if there
are no protected mafiosi. With this secondary utility, after both $a$
and $b$ enter Mafia, $g$ will modify to $\tilde c(a)=\tilde
c(b)=0.5$, giving a Nash equilibrium.

We conjecture that with this secondary utility and the secondary utilities for the vertex agents
as for hitting set, rapid convergence can be shown under an
appropriate choice of the next agent.
}{%

This result is extended to the hitting set game, with a more
sophisticated choice of the next player (see Appendix C). There we
also present a bad example for submodular hitting set.}

\section{Conclusions and further research}\label{sec:concl}
We have defined games whose Nash equilibria correspond to certain
covering problems, with the price of anarchy matching the best
constant factor approximations. The payoffs in these games are locally
defined, and the analysis is based on the LP relaxations of the
corresponding covering problems. An intriguing question is if a
similar game theoretic
approach could be applied for further combinatorial optimization
problems.

The first natural direction would be to extend our approach to a broader class of
covering games.
The most general approximation result on covering games is
\cite{Koufogiannakis09greedy}, giving a
$d$-approximation algorithm for minimizing a
submodular function under monotone constraints,
each constraint dependent on at most $d$ variables. As a first
step, one could study hitting set with the requirement that each
hyperedge $S$ must be covered by at least $h(S)\ge 1$ elements;
a simple primal-dual algorithm was given in \cite{hall86}. However,
extending our game even to this setting does not seem straightforward.

One could also try to formulate analogous settings for classical optimization problems such as
facility location, Steiner-tree or knapsack. One inherent difficulty
is that in our analysis, it seems to be crucial that any greedily
chosen maximal feasible dual solution gives a good
approximation. Also, we heavily rely on the fact  that each constraint
contains at most $d$ variables.

\iftoggle{full}{\smallskip}{}
In Section~\ref{sec:dynamics}, we have shown that the best response
dynamics rapidly converges for vertex cover and hitting set under
certain  assumptions. Stronger convergence results might hold: for
example, it is open if  arbitrary round robin best response dynamics
converge to a Nash equilibrium. For the Submodular Mafia Hitting
  Set Game, we do not even have the weaker convergence result.


\iftoggle{full}{\smallskip}{}

\noindent\textbf{Acknowledgements}
We would like to thank Jarik Ne\v set\v ril for inspiring us to work
on this problem and for a generous support in all directions.


\bibliographystyle{abbrv}
\iftoggle{full}{
\bibliography{Mafia}
}{\bibliography{Mafia-icalp}}

\nottoggle{full}{%
\newpage

\section*{Appendix A: Set cover and hitting set}

We can focus on $d$-uniform hypergraphs without loss of generality due
to the following construction.
Take an arbitrary instance  of the hitting set problem ${\cal G}=(V,{\cal
  E})$, and  let $T>d\sum_{v\in V}c(v)$. Extend $V$ by $d-1$
new vertices of cost $T$, and for every $S\in {\cal E}$, extend $S$ by any
$d-|S|$ new elements. If there is a $d$-approximate solution
to the modified instance, it cannot contain any of the new
elements. Hence finding a $d$-approximate solution is equivalent
in the original and in the modified instance.

In what follows, we present the proof of
Theorem~\ref{thm:hitting-main-icalp}, first proving the existence of a
Nash equilibrium, then bounding the price of anarchy.
The nontrivial new part is Lemma~\ref{lemma:hs-lowdemand} bounding
$D(v)$ for civilians.

The standard LP-relaxation extends the formulations (\ref{prog:PVC})
and (\ref{prog:DVC}).
\begin{multicols}{2}{%
\noindent%
\begin{align}%
\min& \sum_{v\in V} c(v) x(v) \tag{P-HS}\label{prog:PHS}\\%
\sum_{u\in S}x(u)&\ge 1 \quad \forall S\in \C E\notag\\
x&\ge0\notag%
\end{align}
}{%
\begin{align}%
\max& \sum_{S\in\C E} y(S) \tag{D-HS}\label{prog:DHS}\\%
\sum_{S\in \C N(u)} y(S) &\le c(u) \quad \forall u\in V\notag\\
y&\ge0\notag%
\end{align}%
}%
\end{multicols}

Again, for a feasible dual solution $y$, $v\in V$ is called tight if
the corresponding inequality in (\ref{prog:DHS}) holds with equality.
A pair $(M,y)$ of a hitting set $M$ and a feasible dual $y$ is called
a {\em complementary} pair if the dual inequality corresponding to any
$v\in M$ is tight.
The following simple claim generalizes Lemma~\ref{claim:vc-approx}.
\begin{lemma}\label{claim:hs-approx}
If $(M,y)$ is a complementary pair, then $M$ is a $d$-approximate
solution to the hitting set problem.\proofbox
\end{lemma}
The algorithm of Bar-Yehuda and Even \cite{Baryehuda81}, outlined in
Section~\ref{sec:vc} naturally extends to the hitting set
problem, and delivers a complementary pair. Hence the following lemma
proves the existence of a pure Nash equilibrium.

\begin{lemma}\label{lemma:hs-dualne}
Let us define strategies in the Mafia Hitting Set Game based on a
complementary pair $(M,y)$ as follows. Let agents in $M$ be the Mafia and $V\setminus M$ be
the civilians. For each $v\in M$, define $r(v,S)=y(S)$ for every $S\in\C E$
containing $v$.
Then the strategy profile $\C S=(M,C,r)$ is a Nash equilibrium.
\end{lemma}

\begin{proof}
For each $v\in V$,  $D(v)\le \frac{1}{d-1}\sum_{S\in \C
  N(v)}y(S)|(S\cap M)\setminus\{v\}|\le c(v)$ and therefore there are no protected mafiosi.
The proof that nobody has an incentive to change his strategy is the
same as for Lemma~\ref{lemma:vc-dualne}. 
\proofbox\end{proof}

\noindent\textbf{The Price of Anarchy. }
\begin{lemma}\label{lemma:hs-noprotected}
Let the strategy profile $\C S=(M,C,r)$ be a Nash equilibrium.
Then there are no protected mafiosi.
\end{lemma}

\begin{proof}
The proof follows the same lines as for
Lemma~\ref{lemma:vc-noprotected}. For a contradiction, assume $P\neq \emptyset$.
First, it is easy to show that there exists an unprotected $m\in M\setminus P$ and $S\in\C E$, $S\cap P\neq\emptyset$,
such that $r(m,S)>0$ by comparing the total in-demand and
out-demand of protected mafiosi.
For such $m$, if there exists no set $S'\in\C E$ with
 $S'\cap M = \{m\}$, then he could increase his utility by leaving the
 Mafia.
Otherwise, he could increase his utility by decreasing $r(m,S)$ and
increasing $r(m,S')$.
{\proofbox}\end{proof}

The following lemma is the analogue of Lemma~\ref{lemma:vc-lowdemand},
yet the proof is more complicated.
\begin{lemma}\label{lemma:hs-lowdemand}
Let the strategy profile $\C S=(M,C,r)$ be a Nash equilibrium and let $v\in C$.
Then $D(v)\le \frac{d}{d-1} c(v)$.
\end{lemma}

\begin{proof}
Suppose the contrary, let there be a $v\in C$ such that $D(v)>\frac{d}{d-1}c(v)$.
His current utility is $U_{\C S}(v)=-D(v)$.

We show that $v$ could join Mafia and set ransoms that provide him
a strictly larger utility. If $F_{\C S'}^+(v)$ is the income for such a
strategy profile $\C S'$, then $U_{\C S'}(v)=F_{\C S'}^+(v)-2c(v)$,
as he would obtain the protected status. To get $U_{\C S'}(v)>U_{\C
  S}(v)$, we need to ensure $F_{\C S'}^+(v)>2c(v)-D(v)$.
As $D(v)>\frac{d}{d-1} c(v)$ is assumed, it suffices to give an $\C S'$ with
\begin{equation}\label{eq:income-bound}
F_{\C S'}^+(v)\ge \frac{d-2}{d-1}c(v).
\end{equation}
We define the ransoms $r'(v,S)$ by ``stealing'' the strategies of the other
mafiosi.
That is, for each club $S\in \C N(v)$, $\frac{1}{d-1}\sum_{m\in M\cap
  S}r(m,S)$ is the total ransom $v$ has to pay to the members of this
club. We define $r'(v,S)$ proportionally to this amount:
\[
r'(v,S):=\frac{c(v)}{(d-1)D(v)}\sum_{m\in M\cap   S}r(m,S).
\]
By Lemma~\ref{lemma:hs-noprotected}, we know that there are
no protected mafiosi in the Nash equilibrium $\C S$.
We show that after $v$ enters Mafia, even if some of the old mafiosi
become protected, they are only slightly overcharged.
More precisely, we shall show that
\begin{equation}\label{eq:overcharge}
D'(t)\le \frac{d}{d-1} c(t)\ \ \forall t\in M.
\end{equation}
From this bound, (\ref{eq:income-bound}) immediately follows.
Indeed, everybody will pay at least $\frac{d}{d-1}$ fraction of
  the demands, and therefore $F_{\C S'}^+(v)\ge \frac{d-1}{d}c(v)\ge \frac{d-2}{d-1}c(v)$.

It is left to prove (\ref{eq:overcharge}). The demand of $v$ from some
$t\in M$ can be bounded as follows:
\begin{align*}
&\frac{1}{d-1}  \sum_{S\in \C N(v)\cap \C N(t)} r'(v,S) = \\
& = \frac{c(v)}{(d-1)^2D(v)}  \sum_{S\in
  \C N(v)\cap\C N(t)}
\left( r(t,S) + 
\sum_{t'\in (S\cap M)\setminus\{t\}} r(t',S) \right)\\
&  < \frac{1}{d(d-1)} (c(t)+ (d-1)D(t))  \le \frac{1}{d(d-1)} d\cdot c(t) = \frac{c(t)}{d-1}.
\end{align*}
Here we used that $D(t)\le c(t)$ as $t$ was not
protected in $\C S$. Using this fact once more, we get
\[
D'(t) \le D(t) + \frac{c(t)}{d-1} \le \frac{d}{d-1} c(t).
\]%
\nottoggle{full}{\proofbox}{}\end{proof}

\begin{theorem}
The price of anarchy for the Mafia Hitting Set Game is $d$.
\end{theorem}

\begin{proof}
Let $\C S=(M,C,r)$ be a strategy profile in a Nash equilibrium.
Then $M$ is a hitting set, as if there was an uncovered club, all members
would be unhappy due to the term $\operatorname{Pen}(v)$. We show that the
cost of $M$ is within a factor $d$ from the optimum.
Let us set $y(S)=\sum_{v\in M\cap S} r(v,S)$ for each $S\in\C E$.
Lemmas~\ref{lemma:hs-noprotected} and \ref{lemma:hs-lowdemand} easily
imply
$\sum_{S\in\C S: v\in S} y(S)\le d\cdot c(v)$ for every $v\in V$, and
thus $\frac 1d y$ is a feasible dual solution to (\ref{prog:DHS}).
Then
$$
\sum_{S\in\C S} \frac{1}{d} y(S) = \sum_{s\in \C S} \frac{1}{d} \sum_{m\in M\cap S}r(m,S)
=\sum_{m\in M} \frac{1}{d} \sum_{S\in \C N(m)} r(m,S) = \frac 1d \sum_{m\in M} c(m),
$$
showing that $M$ is a $d$-approximate solution to (\ref{prog:PHS}).\proofbox
\end{proof}

\section*{Appendix B: Submodular hitting set}
We shall present the proof of Theorem~\ref{thm:submod-main-icalp},
first showing the existence of a Nash equilibrium.
Besides the submodular base polyhedron defined in Section~\ref{sec:submod}, we
also need the notion of the
{\em submodular polyhedron} 
\[
P(C)=\{z\in \R^V: z\ge 0, z(Z)\le C(Z)\ \ \forall Z\subseteq V\}.
\]
Given a vector $z\in P(C)$, the set $Z$ is {\em tight} with respect to
$z$ if $z(Z)=C(Z)$. An elementary consequence of submodularity is that
for every $z\in P(C)$, there exists a unique maximal tight set.
Note that $B(C)\subseteq P(C)$ and $z\in P(C)$ is in $B(C)$ if and
only if $V$ is tight.

In the LP relaxation, we assign a primal variable $\xi(Z)$ to every
subset $Z\subseteq V$. In an integer solution, $\xi(Z)= 1$ if $Z$ is the
chosen hitting set and 0 otherwise.
\begin{multicols}{2}{%
\noindent%
\begin{align}
\min& \sum_{Z\subseteq V} C(Z) \xi(Z) \tag{P-SHS}\label{prog:PSHS}\\
\sum_{Z\in \C N(u)}\xi(Z)&=x(u) \quad\forall u\in V\notag\\
\sum_{u\in S}x(u)&\ge 1 \quad \forall S\in \C E \notag\\
\xi&\ge0\notag
\end{align}%
}{%
\begin{align}
\max& \sum_{S\in\C E} y(S) \tag{D-SHS}\label{prog:DSHS}\\
\sum_{S\in \C N(u)} y(S) &= z(u) \quad \forall u\in V\notag\\
z&\in P(C)\notag\\
y&\ge0\notag
\end{align}%
}%
\end{multicols}

Note that in the dual program, $y$ uniquely defines $z$. Therefore we
will say that $y$ is a feasible dual solution if the corresponding $z$
is in $P(C)$. For the special case of the (linear)
hitting set problem, where $C(Z)=\sum_{v\in Z}c(z)$ for some
$c:V\rightarrow\R_+$, this is equivalent to $y$ satisfying (\ref{prog:DHS}).

Accordingly, we say that a set $Z$ is tight for a
feasible dual $y$ if $z(Z)=P(C)$.
For a hitting set $M$ and a feasible dual solution $y$, we say that
$(M,y)$ is a {\em complementary pair} if $M$ is tight for $y$. The
following is the
generalization of Lemmas~\ref{claim:vc-approx} and \ref{claim:hs-approx}.
\begin{lemma}\label{claim:shs-approx}
If $(M,y)$ is a complementary pair, then $M$ is a $d$-approximate
solution of the hitting set problem.
\end{lemma}
\begin{proof}
The primal objective is at most $d$ times the dual objective, as
\[
C(M)=\sum_{v\in M}z(m)=\sum_{v\in M}\sum_{S\in \C N(v)}y(S)\le d\sum_{S}y(S).
\]
The inequality follows as each $S$ is counted $|S|\le d$ times.
\nottoggle{full}{\proofbox}{}\end{proof}

The algorithm by Iwata and Nagano \cite{Iwata09} is as follows.
\begin{enumerate}
\setcounter{enumi}{-1}
\item Set $y(S):=0$ for each $S\in \C E$, $z(v):=0$ for $v\in V$,  and let $M$ be the unique maximal
  set with $C(M)=0$.
\item While $M$ is not a hitting set do
\begin{enumerate}
\item Choose an arbitrary hyperedge $S\in \cal E$, $S\cap M=\emptyset$.
\item\label{step:epsilon} Compute $\varepsilon=\max\{\lambda:z+\lambda\chi_Z\in P(C)\}.$
\item Increase $y(S)$ and every $z(v)$ for $v\in Z$ by $\varepsilon$.
\item Replace $M$ by the new unique maximal tight set.
\end{enumerate}
\item Return $M$.
\end{enumerate}

In step (1-2), $\chi_Z$ is the characteristic function of $Z$. This
step can be performed in the same running time as a submodular
function minimization (see \cite{fleischer03}). Note also that $M$ will always intersect $S$ in
step (1-4) and therefore will be strictly extended. It is immediate that it
returns a complementary pair $(M,y)$ and thus
Lemma~\ref{claim:shs-approx} proves $d$-approximation.

As for vertex cover and hitting set, we show that a complementary
solution $(M,y)$ to (\ref{prog:PSHS}) and (\ref{prog:DSHS}) provides a
solution in Nash equilibrium. Let $z(u)=\sum_{S\in \C N(u)} y(S)$.
Note that $z\in P(C)$ and $M$ is tight for $z$.
Let us raise the $z(v)$ values for $v\in C$ arbitrarily in order to get
a vector in the base polyhedron $B(C)$. Let $\tilde c$ denote such a vector.

\begin{lemma}\label{lemma:shs-dualne}
Let us define strategies in the Mafia Hitting Set Game based on a
complementary pair $(M,y)$ with $M$ being the Mafia and $V\setminus M$ the civilians. Let the
Godfather assign the budget vector $\tilde c$ as defined above.
For $u\in M$ and $S\in \C N(u)$, define $r_0(u,S)=y(S)/\tilde c(u)$.
Then the strategy profile $\C S=(M,C,\tilde c,r_0)$ is a Nash equilibrium.
\end{lemma}
\begin{proof}
The Godfather has no incentive to change as by $\tilde c(M)=C(M)$, he
already receives the maximum possible utility for the given $M$.
By the definition, $r(u,S)=y(S)$ for each $v\in V$,  hence $D(v)\le
\frac{1}{d-1}\sum_{S\in \C N(v)}y(S)|(S\cap M)\setminus\{v\}|\le
\tilde c(v)$ and therefore there are no protected mafiosi.
The proof that nobody has an incentive to change his strategy is the
same as for Lemma~\ref{lemma:vc-dualne}.
{\proofbox}{}\end{proof}

\begin{theorem}
The price of anarchy for the Submodular Mafia Hitting Set Game is $d$.
\end{theorem}
\begin{proof}
Consider a strategy profile $\C S=(M,C,\tilde c,r_0)$ in a Nash equilibrium.
We can repeat the entire argument of Section~\ref{sec:hs} to show
that there are no protected mafiosi and that every civilian is
demanded at most $\frac{d}{d-1} \tilde c$. This is since if the Godfather
does not change his strategy, the game is identical to the linear game
with fixed budgets $\tilde c$ from the perspective of the vertex agents.

Finally we define $y(S)=\sum_{m\in M\cap S}r(m,S)$, and then $\frac 1d y$
is a feasible dual solution for (\ref{prog:DSHS}).
Observe that
$\tilde c(M)=C(M)$, as otherwise the Godfather would have incentive to
decrease the budgets of certain civilians and increase for certain
mafiosi.
Consequently,
$C(M)=\tilde c(M)=\sum_{S\in \C E}y(S)$, showing that $M$ is a
$d$-approximate solution for (\ref{prog:PSHS}).\proofbox
\end{proof}

\section*{Appendix C: Convergence to Nash equilibrium}
\subsection*{Vertex cover}
We first present an example showing that a simple round robin dynamics
may run into loop without the secondary objective. Then we give the
proof of Theorem~\ref{thm:vc-dynamics}.

\medskip

{\noindent\bf Example. }Consider a star on the vertices $v_1,v_2,v_3,v_4$ and the central vertex
$z$.
Assume we are playing round robin in the order $z,v_1,v_2,v_3,v_4$.
Let $c(v_i)=1$ for $i=1,2,3$ and $c(z)=2$, and let us start with the
strategy profile where  $M=\{z\}$, $r(z,v_1)=2$.
Assume that whenever $z$ can change his strategy to get a higher
utility, he always chooses to demand his entire budget 2 from one of
the civilians among $v_1,v_2,v_3,v_4$ (this is always a best response).

We claim that this will always be possible as $z$ always stays in the
Mafia, and at most 3 vertices among
$v_1,v_2,v_3,v_4$ will be in the Mafia at the same time.
Indeed, a civilian will enter only if being ransomed by
$z$. If $v_i$ is in the Mafia then his only option is
setting $r(v_i,z)=1$, and thus if $z$ has at least 3 neighbors in the
Mafia, he becomes protected and thus all his neighbors he is not
actually ransoming will have an
incentive to leave.

The dynamics never reaches a Nash equilibrium, as if $z$ is ransoming a
mafioso $v_i$,  he has incentive to change to ransoming a civilian as
$v_i$ is protected. On the other hand,  if $z$ ransoms a civilian $v_i$, $v_i$
has an incentive to join the
Mafia to obtain the protected status.

\medskip

\noindent{\em Proof of Theorem~\ref{thm:vc-dynamics}.}
 By induction, we shall prove that in every
step, $c(v)\ge D(v)$ and $\tilde U_{\cal S}(v)=0$ for all $v\in V$.
Consider the next move, when a player $v$ incident to some uncovered
edges minimizing $c(v)-D(v)$ moves.
He obviously has to enter the Mafia, and can achieve a maximal
(primary and secondary) utility if he sets $r(v,u)=r(u,v)$ for any
$u\in M\cap N(v)$, and distributes the rest of his ransoms arbitrarily to his
civilian neighbors. Note that this can always be done because
$c(v)\ge D(v)$. Also, note that the total ransom $v$ will demand from
other civilians is $c(v)- D(v)$. By the extremal choice of $v$, it
follows that none of his civilian neighbors $z$ will violate $c(z)\ge
D(z)$. This also remains true if $z\in M$, as $D(z)$ is at most
the total ransom $z$ demands due to the symmetry of the ransoms.

Hence the induction hypothesis is maintained by an arbitrary best
response of $v$. A mafioso who is not protected and has secondary
objective 0 has no incentive to change his strategy. Also, a civilian
$v$ with $c(v)\ge D(v)$ has no incentive to join the Mafia if there
are no uncovered edges incident to $v$. Consequently, the game ends
after all uncovered edges are gone, and once an agent joins to Mafia, he
would not change his strategy anymore.
\proofbox

\subsection*{Hitting set}
The natural generalization of the secondary objective for hitting set
is as follows. For a club $S\in \C E$, let $var(S)$ denote the maximum
difference between ransoms on this edge. That is, $var(S)=0$ if
$|S\cap M|\le 1$ and $var(S)=\max_{v\in S\cap M}r(v,S)-\min_{v\in
  S\cap M}r(v,S)$ otherwise.
For a strategy profile $\C S=(M,C,r)$, let $\tilde U_{\C S}(v)=0$ if
$v\in C$ and
\[
\tilde U_{\C S}(v)=-\sum_{v\in \C N(v)}var(S)
\]
if $v\in M$. The utility of an agent is then $(U_{\C S}(v),\tilde
U_{\C S}(v))$, under lexicographic ordering.

A natural expectation would be to prove rapid convergence as for
vertex cover, if always the agent minimizing $c(v)-D(v)$  is allowed
to play. However, the Bar-Yehuda--Even algorithm does not seem to be
modeled by this dynamics.
Instead, we define a slightly different extremal choice of the next
agent. Let
\[D^*(v)=\sum_{S\in \C N(v)} \max_{m\in S\cap
  M}r(m,S),\]
 that is, for each club $S$ we consider the largest ransom
demanded in this club. Note that $D(v)\le D^*(v)$.
Let us consider the following dynamics. {\em We start from the strategy
profile where everyone is civilian, and we always let a civilian play
next who is contained in an uncovered club. Among them, we let
the one play who minimizes $c(v)-D^*(v)$.}

\begin{theorem}\label{thm:hs-dynamics}
After each agent changing his strategy at most once, we obtain a
strategy profile in Nash equilibrium.
\end{theorem}
\begin{proof}
We prove by induction, that in every step,
 $c(v)\ge D^*(v)$ and $\tilde U_{\cal S}(v)=0$ for all $v\in V$.
Note that this implies that there are no protected mafiosi.
If $M$ is not a hitting set, we let a $v$ minimizing $c(v)-D^*(v)$ play.
$\tilde U_{\cal S}(m)=0$ for all $m\in M$ means that for every club $S$,
$r(m,S)$ is equal for every $m\in M\cap S$; let $r_S$ denote this
common value.
As for the vertex cover case, the best responses of $v$ are to set
$r(v,S)=r_S$ whenever $S$ was already covered by the Mafia, and
to distribute the remaining ransoms arbitrarily on the hyperedges
covered only by $v$.

As $D^*(v)=\sum_{S\in \C N(v)}r_S$, the remaining
amount $v$ distributes is exactly $c(v)-D^*(v)$. Then by the choice of
$v$, $c(z)\ge D^*(z)$ shall be maintained for every civilian $z$, and
also for other mafiosi (note that if $z\in M$, then $D^*(z)$ does not change).
It can be seen analogously as for vertex cover, that we have a Nash
equilibrium if there are no more
uncovered clubs.
\proofbox\end{proof}

Similarly to the vertex cover case, this dynamics essentially simulates the
Bar-Yehuda--Even algorithm. Also, an analogous distributed interpretation can be
given.

\subsection*{Submodular hitting set}
One would expect that the Submodular Mafia Hitting Set Game also
converges under some dynamics that simulates the primal-dual algorithm
by Iwata and Nagano \cite{Iwata09}. However, if the Godfather does not
have a secondary utility, the following example shows that it can
run into a loop even in very simple instances.

Let $V=\{a,b\}$, $C(\{a\})=C(\{b\})=C(\{a,b\})=1$, and let $g$ be the
Godfather. Let $\C E=\{\{a,b\}\}$; for simplicity, we use the notation
of vertex cover, e.g. $r_0(a,b)$ denotes $r_0(a,\{a,b\})$.
Let us start from the strategy profile $C=V$, $\tilde c(a)=1$, $\tilde
c(b)=0$, and play a round robin in the order $a,b,g$.

First, $a$ enters Mafia and sets $r_0(a,b)=1$. Then $b$ also enters to
receive the protected status and sets $r_0(b,a)=1$.
$g$ has no incentive to move as $\tilde
c(M)=1$ is already maximal. In the next round, $a$ is happier if he
leaves Mafia; $b$ has no incentive to change, however $g$
modifies to $\tilde c(a)=0$ and $\tilde c(b)=1$. This will lead to a
loop: $a$ enters again in next round, $b$ leaves, $\tilde c$ is
changed again, etc.

The above behavior can be avoided by introducing a secondary
utility for $g$: let $\tilde U_{\C S}(g)=\sum_{v\in M}F^+(v)$, that
is, the sum of the actual incomes of the mafiosi. Note that $\tilde
U_{\C S}(g)\le \tilde c(M)$ and equality holds if and only if there
are no protected mafiosi. With this secondary utility, after both $a$
and $b$ enter Mafia, $g$ will modify to $\tilde c(a)=\tilde
c(b)=0.5$, giving a Nash equilibrium.

We conjecture that with this secondary utility and the secondary utilities for the vertex agents
as for hitting set, rapid convergence can be shown under an
appropriate choice of the next agent.
}{}

\end{document}